\documentclass[11pt]{article} 

\usepackage[utf8]{inputenc} 

\usepackage{hyperref}
\usepackage{cite}
\hypersetup{colorlinks=false}

\usepackage{geometry,amsmath} 
\numberwithin{equation}{section}

\usepackage{booktabs} 
\usepackage{array} 
\usepackage{paralist} 
\usepackage{verbatim} 

\usepackage{todonotes}
\usepackage[subrefformat=parens]{subcaption}
\usepackage{textcomp}
\usepackage{float}
\usepackage{amssymb}
\usepackage{amsmath}
\usepackage{amsthm}
\usepackage{mathrsfs}
\usepackage{mathtools}
\usepackage{bm}

\usepackage{xcolor,pict2e}

\usepackage{fancyhdr} 
\usepackage{sectsty}
\allsectionsfont{\sffamily\mdseries\upshape} 
\usepackage{lipsum}

\usepackage{caption}
\theoremstyle{definition}

\newtheorem{thm}{Theorem}

\newtheorem{prop}[thm]{Proposition}

\newtheorem{defn}[thm]{Definition}
\newtheorem{rem}[thm]{Remark}

\newcommand{\Figref}[1]{Fig.~\ref{#1}}

\newcommand{\norm}[1]{\left\lVert#1\right\rVert}

\newcommand{\Log}[1]{\operatorname{Log} #1 }
\newcommand{\Exp}[1]{\operatorname{Exp} #1 }
\newcommand{\tr}{\operatorname{tr} }
\newcommand{\argmin}{\operatorname{arg}\min }
\newcommand{\hermconj}{^{\mathsf{H}}}
\newcommand{\frechetderi}{  \left.\frac{d}{d\varepsilon}\right|_{\varepsilon = 0}}

\begin{document}

\title{Towards a median signal detector through the total Bregman divergence and its robustness analysis }
\author{Yusuke Ono\thanks{E-mail: yuu555yuu@keio.jp} \ and Linyu Peng\thanks{Corresponding author. E-mail: l.peng@mech.keio.ac.jp} 
\vspace{0.4cm}
\\
{\it Department of Mechanical Engineering, Keio University,}
\\
{\it Yokohama 223-8522, Japan}}

%
\maketitle

\abstract{%
A novel family of geometric signal detectors are proposed through  medians of the total Bregman divergence (TBD), which are shown advantageous over the conventional methods and their mean counterparts.  By interpreting the observation data as Hermitian positive-definite matrices, their mean or median play an essential role in signal detection. As is difficult to be solved analytically, we propose numerical solutions through Riemannian gradient descent algorithms or fixed-point algorithms. 

Beside  detection performance, robustness of a detector to outliers is also of vital importance, which can often be analyzed via the influence functions.  Introducing an orthogonal basis for Hermitian matrices, we are able to compute the corresponding influence functions analytically and exactly by solving a linear system, which is transformed from the governing matrix equation. Numerical simulations show that the TBD medians are more robust than their mean counterparts. 

{\bf  Keywords:} Geometric median, matrix-CFAR, robustness analysis, nonhomogeneous clutter 
}

\section{Introduction} 
In target detection, the constant false alarm rate (CFAR) is a well-known and conventionally used clutter suppression method \cite{RadarSignalProcessing}. A popular CFAR is based on Fourier transform and makes use of the Doppler spectral density to distinguish targets from clutter. However, this method is rather weak for short waveform of dense and nonhomogeneous clutter. 
One of the popular statistical methods for discriminating between the target signal from clutter is to estimate the covariance matrix and test the likelihood ratio. 
There has been a well-known method for estimating the clutter covariance matrix, namely, the sample covariance matrix (SCM) using the maximum likelihood estimation \cite{SCMbyMLE}, but this method performs well only when  enough number of signals is independent and identically distributed \cite{4101326}.
In order to reduce the number of required observation data for achieving efficient detection performance within nonhomogeneous clutter, various knowledge-based methods using {\it a priori} information were proposed \cite{5417154,LI20191,s21072391,6853408,ctx11663242640004034}. 
In \cite{5484507, 4154721}, the Bayesian methods were leveraged to estimate the covariance matrix.
In practice, it is hard to accurately capture the statistical information of clutter in advance which can be range-dependent and non-stationary and the observation data is usually contaminated. 
Insufficient knowledge of clutter environment can lead to the detection performance degradation.

Recently, the matrix-CFAR, which needs no {\it a priori} information but makes use of the covariance matrix of the signals, was proposed and developed \cite{BarInoTool,ETIVC2009}. 
It is based on the fact that the covariance matrix that represents the correlation of complex signals is Hermitian positive-definite (HPD). 
It was applied to target detection \cite{7560351} and drone detection \cite{DroneDetect}. 
Its advantages have been demonstrated in detection of high frequency X-band radar clutter \cite{BarHFandXband}, Burg estimation of radar scatter matrices \cite{7842633}, 
monitoring of wake eddy turbulence \cite{BARBARESCO201054}, etc. 
In recent years, the matrix-CFAR with respect to various divergence functions has been proposed, and its detection performance and robustness to outliers were shown better compared with that of the geodesic distance associated to the affine-invariant Riemannian metric (AIRM) \cite{fixpointHua,MatrixCFARbyKL,CFARJensen,9764734}. 

 In \cite{myarticle}, the total Bregman divergence (TBD) defined in  HPD manifolds was applied to signal detection in nonhomogeneous clutter. The TBD means associated to the total square loss (TSL), the total log-determinant (TLD) divergence and the total von-Neumann (TVN) divergence were investigated; it was shown numerically that the corresponding detectors outperform the AIRM detector and generalized likelihood ratio test (GLRT) detector. 
In the current study, to improve the robustness of the matrix-CFAR, we extend the previous study \cite{myarticle} on TBD means to TBD medians as median detectors are often more robust to outliers than the mean counterparts. In addition, a novel analytical method for robustness analysis is proposed.
Main contributions of the current paper are outlined as follows. 
\begin{itemize}
\item[(1)]  We define the TBD medians associated with the TSL, the TLD divergence and the TVN divergence, and investigate their detection performance compared with the TBD means, the Riemannian distance (RD) mean and median associated to the AIRM,  the GLRT detector as well as the adaptive normalized matched filter (ANMF). Medians are often difficult to be calculated analytically; instead, we introduce fixed-point algorithms or Riemannian gradient descent algorithms to solve them numerically. 
\item[(2)] Influence functions can be used to describe the robustness of a detector about outliers. As noticed in \cite{myarticle}, the matrix equations for influence functions are not commonly seen in matrix theory and to solve them directly can be very challenging;  approximate solutions were obtained and applied therein.  By defining an orthogonal basis for Hermitian matrix spaces with the induced Frobenius metric in this study,  we
transform those matrix equations to linear systems, which can immediately (at least in principle) be  solved exactly.  This methodology allows us to precisely capture the robustness of all means and medians discussed in the current paper. 
\end{itemize}

The structure of this paper is as follows. Theory of the matrix-CFAR is summarized in Section \ref{sec:mat}.
In Section \ref{sec:HPD}, after briefly reviewing the geometry of HPD equipped with the AIRM, we introduce Riemannian gradient descent algorithms for solving the corresponding RD mean and median. Riemannian gradient, also called natural gradient, assigns the steepest direction in Riemannian manifolds. Furthermore, the TBD medians are defined and  solved numerically through the  fixed-point algorithms.
Numerical detection performance is conduced in Section \ref{sec:DP} by comparing the RD mean and median, the TBD means and medians, and the conventional GLRT. 
In Section \ref{sec:RA}, an orthogonal basis for Hermitian matrices is defined and applied to the computation of influence functions, allowing us to accurately illustrate robustness of all means and medians.
We conclude and address  several lines of potential future researches finally in Section \ref{sec:con}.

\section{Formulation of the matrix-CFAR}
\label{sec:mat}
In this section, the method of matrix-CFAR will be recalled. 
Assume that the observation data are collected from $N$ stationary channels and
consider a detection problem concerning a moving target under the nonhomogeneous clutter environment. The problem of detection  will be modeled as a binary hypothesis testing. In this testing, the null hypothesis $H_0$ represents that the observation data $\bm{x}$ is only the clutter $\bm{c}$ and the alternative hypothesis $H_0$ represents that $\bm{x}$ contains both the target signal $\bm{s}$ and the clutter $\bm{c}$, namely
\begin{eqnarray}
    \begin{cases}
        H_0:
        \begin{cases}
			\mathrm{target\ is\ absent,} \\
            \bm{x} = \bm{c}, \\
        \end{cases}  
            \vspace{0.2cm}\\
        H_1:
        \begin{cases}
			\mathrm{target\ is\ present,} \\
            \bm{x} = \xi  \bm{s} + \bm{c}, \\
        \end{cases}  \\
    \end{cases}
\end{eqnarray}
where $\xi$ denotes the unknown amplitude coefficient which is  complex scalar-valued. 
The observation data $\bm{x}$ and the target signal $\bm{s}$ denoting the steering vector are respectively given by
\begin{eqnarray}
        &\bm{x} = \left(x_0, \ldots, x_{N-1}\right)^{\operatorname{T}}\in\mathbb{C}^N, \label{eq: recieve signal}\\
        &\bm{s} = \dfrac{1}{\sqrt{N}}\left(1, \exp(-\operatorname{i}2\pi f_d), \ldots, \exp(-\operatorname{i}2\pi f_d(N-1))\right)^{\operatorname{T}}, \label{eq:target}
\end{eqnarray}
where $\mathbb{C}^N$ represents the $N$ dimensional complex space, $\operatorname{i}$ denotes the imaginary unit, $f_d$ is the signal normalized Doppler frequency, and $(\cdot)^{\operatorname{T}}$ denotes the transpose of matrices or vectors.
Conventional SCM estimator, derived by secondary data set $\left\{ \bm{x}_i \right\}_{i=1}^m$, is given by \cite{SCMbyMLE} 
\begin{eqnarray}
    \label{eq:SCM}
	\bm{R}_{\mathrm{SCM}} = \frac1m \sum_{i=1}^m  \bm{x}_i \bm{x}_i\hermconj,\ \  \bm{x}_i \in \mathbb{C}^N ,
\end{eqnarray}
where $(\cdot)\hermconj$ represents the conjugate transpose of matrices or vectors. This  estimator has been widely utilized;  for instance, the GLRT detector reads \cite{GLRT}
\begin{eqnarray}
    \frac{\left| \bm{x}\hermconj \bm{R}^{-1}_{\mathrm{SCM}} \bm{s} \right|^2}{\bm{s}\hermconj\bm{R}^{-1}_{\mathrm{SCM}} \bm{s} } \overset{H_1}{\underset{H_0}{\gtrless}} \gamma 
\end{eqnarray}
and the ANMF reads \cite{381910}
\begin{eqnarray}
    \frac{\left| \bm{s}\hermconj \bm{R}^{-1}_{\mathrm{SCM}} \bm{x} \right|^2}{\left( \bm{x}\hermconj\bm{R}^{-1}_{\mathrm{SCM}} \bm{x} \right)\left( \bm{s}\hermconj\bm{R}^{-1}_{\mathrm{SCM}} \bm{s} \right)} \overset{H_1}{\underset{H_0}{\gtrless}} \gamma, 
\end{eqnarray}
where $\gamma$ denotes the detection threshold. 

In the matrix-CFAR, to determine whether the received signal $\bm{x}$ includes a target $\bm{s}$ or not, we model the observation data by the covariance matrix as a Toeplitz HPD matrix defined as \cite{6514112,40004034}
\begin{equation}\label{eq:covariancematrix}
	\begin{aligned}
	\hspace{-7mm}\bm{R} = 
	    \begin{pmatrix}
		    r_0 & \cdots & \overline{r}_k & \cdots & \overline{r}_{N-1} \\
		    \vdots & \ddots & \ddots & \ddots & \vdots \\
		    r_k & \ddots & r_0 & \ddots & \overline{r}_k \\
		    \vdots & \ddots & \ddots & \ddots & \vdots \\
		    r_{N-1} & \cdots & r_k & \cdots & r_0 \\
	    \end{pmatrix}
	    \in \mathscr{P}(N,\mathbb{C}) ,\\
	\end{aligned}
\end{equation}
where $ \mathscr{P}(N,\mathbb{C}) $ denotes the space of $N \times N $ HPD matrices. Its component
\begin{equation}
		r_k = \mathbb{E} \left[x_l \overline{x}_{l+k} \right], \ \ 0 \le k \le N-1, \ 0 \le l \le N-l-1 
\end{equation}
is the $k$-th correlation coefficient.  Here, $\overline{r}_{k}$ is the conjugate of $r_{k}$, and $\mathbb{E}[\cdot]$ is the statistical expectation. The Toeplitz covariance structure has been widely studied and applied as it can often improve the estimate of the covariance matrix; see, e.g., \cite{ABY2012,9054969,987665}.


Based on the ergodicity of a stationary process, each component $r_{l}$ can be approximated by the observation data as
\begin{eqnarray}
	\label{eq:autoregressive}
		r_k & = \dfrac{1}{N} \sum\limits_{l = 0} ^{N-1-l} x_l \overline{x}_{l+k} ,\ \ 0 \le k \le N-1. 
\end{eqnarray}
The covariance matrix of clutter is estimated as $\bm{R}_{g}$ by the observations $\{ \bm{R}_i\}_{i=1}^m$. Consequently,  the detection problem can be modeled as 
\begin{eqnarray}
    \begin{cases}
        H_0:\bm{R}_{\mathrm{CUT}} = \bm{R}_{g}, \vspace{0.15cm} \\
        H_1:\bm{R}_{\mathrm{CUT}} \neq \bm{R}_{g}, \\
    \end{cases}
\end{eqnarray}
where the matrix $\bm{R}_{\mathrm{CUT}}$ is the covariance matrix of the observation in the cell under test (CUT). By defining a threshold $\gamma$, the signal detection is modeled as discriminating $\bm{R}_{\mathrm{CUT}}$ from $\bm{R}_{g}$ by
\begin{eqnarray}\label{eq:decison_of_R_g}
    d\left( \bm{R}_g, \bm{R}_{\mathrm{CUT}}\right) \overset{H_1}{\underset{H_0}{\gtrless}} \gamma, 
\end{eqnarray}
where $d\left( \bm{R}_{g}, \bm{R}_{\mathrm{CUT}}\right)$ is the difference between $\bm{R}_{\mathrm{CUT}}$ and $\bm{R}_{g}$, such as the Riemannian distance or TBD divergences that we are going to introduce in the next section. 

The matrix-CFAR illustrated by \Figref{fig:processMCFAR} is summarized up as follows. Estimator of covariance matrices $\{ \bm{R}_i\}_{i=1}^m$ of the observation data $\{ \bm{x}_i\}_{i=1}^m$ is firstly calculated by \eqref{eq:covariancematrix} and \eqref{eq:autoregressive}. 
Mean or median $\bm{R}_g$ of those covariance matrices will be obtained by using either the distance or divergences. The matrix $\bm{R}_g$ can be considered as an estimator of the covariance matrix of the clutter since we consider almost all observation data $\{ \bm{x}_i\}_{i=1}^m$ as the clutter, that is essential in design of the matrix-CFAR. 
The  detection decision is performed by the comparison of the distance or difference of $\bm{R}_{\mathrm{CUT}}$ and $\bm{R}_g$ and the threshold $\gamma$. 
\begin{figure}[htbp]
	\centering
	\includegraphics[width = 0.8\linewidth]{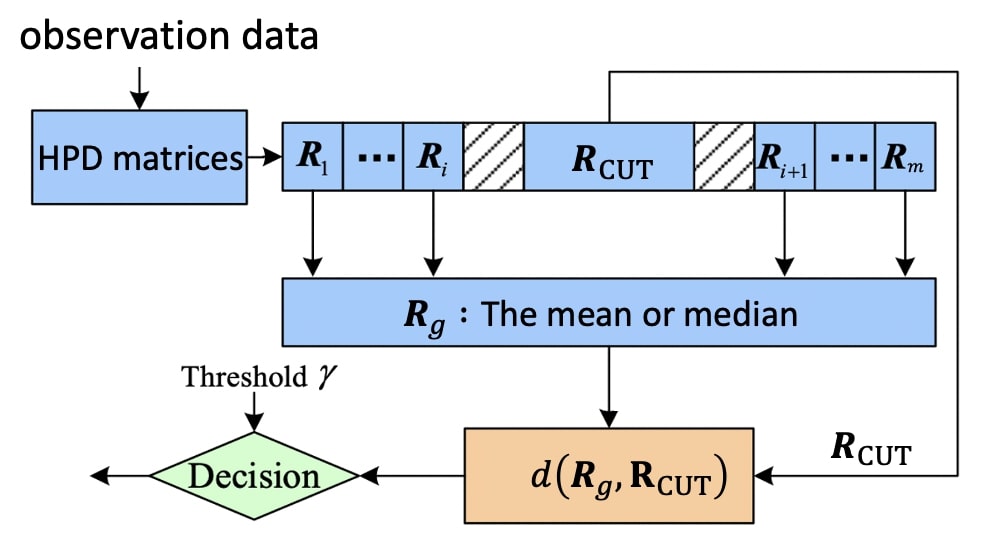}
	\caption{The process of matrix-CFAR.}
	\label{fig:processMCFAR}
\end{figure}

\section{Means and medians of HPD matrices}
\label{sec:HPD}
In this section, we briefly introduce geometric means and medians in HPD spaces equipped with either the AIRM as Riemannian manifolds or  the Frobenius metric as metric spaces. In the latter, divergence functions can be defined, playing the role of measuring the difference of two HPD matrices. 
\subsection{HPD manifolds and the RD mean and median}
We firstly define the AIRM on the HPD matrix manifold $\mathscr{P}(N,\mathbb{C})$, and introduce the corresponding geodesic distance, also called the Riemannian distance. The tangent space at a point  $\bm{P}\in \mathscr{P}(N,\mathbb{C})$ is  $T_{\bm{P}}\mathscr{P}(N,\mathbb{C}) \cong \mathscr{H}(N,\mathbb{C})$, which denotes the space of $N\times N$ Hermitian matrices. The tangent bundle is denoted by $T\mathscr{P}(N,\mathbb{C})=\cup_{\bm{P}}T_{\bm{P}}\mathscr{P}(N,\mathbb{C})$.
The well-known AIRM is defined by
\begin{eqnarray}
\langle \bm{A},\bm{B} \rangle_{\bm{P}} := \tr \left(\bm{P}^{-1}\bm{A}\bm{P}^{-1}\bm{B}\right) 
\end{eqnarray}
for  $ \bm{A},\bm{B}\in T_{\bm{P}}\mathscr{P}(N,\mathbb{C})$, where $\operatorname{tr}$ denotes the trace of a matrix.
The geodesic curve between two points $\bm{P}_0, \bm{P}_1 \in \mathscr{P}(N,\mathbb{C}) $ is given by
\begin{equation}\label{eq:geocurve}
    \bm{P}(t) = \bm{P}_0^{\frac12} \exp \left( t \bm{P}_0^{\frac12} \Log \left( \bm{P}_0^{-\frac12} \bm{P}_1 \bm{P}_0^{-\frac12} \right)\bm{P}_0^{\frac12} \right) \bm{P}_0^{\frac12}, \ t\in[0,1],
\end{equation}
where $\exp$ denotes matrix exponential and $\operatorname{Log}$ denotes the principal logarithm \cite{Hig2008}.
Obviously, $\bm{P}(0)  = \bm{P}_0 ,$  and $\bm{P}(1) = \bm{P}_1 $.
Under the AIRM metric, the geodesic distance between two points $\bm{P}_0, \bm{P}_1 \in \mathscr{P}(N,\mathbb{C}) $ is given by 
\begin{eqnarray}\label{eq:Riedis}
		d_R\left( \bm{P}_0, \bm{P}_1 \right) 
		& = &\norm{ \Log \left( \bm{P}_0^{-\frac12} \bm{P}_1 \bm{P}_0^{-\frac12} \right) } \nonumber \\
		& = &\norm{ \Log \left(  \bm{P}_1^{-1} \bm{P}_0 \right) }, 
\end{eqnarray}
where $\left\| \cdot \right\|$ is the Frobenius norm associated to the Frobenius metric $\left\langle \bm{A}, \bm{B} \right\rangle := \tr{\left( \bm{A}\hermconj \bm{B}\right)}$. For more details, the reader may refer to \cite{RiegeoMoa}.

In the Riemannian manifold  $\mathscr{P}(N,\mathbb{C})$ equipped with the AIRM, the exponential map 
$\Exp: T\mathscr{P}(N,\mathbb{C})  \rightarrow \mathscr{P}(N,\mathbb{C})$
is defined by making use of the geodesics. Locally, assuming $\bm{P}(t)$ is the geodesic  \eqref{eq:geocurve} connecting two arbitrary HPD matrices $\bm{P}_0=\bm{P}(0),\bm{P}_1=\bm{P}(1)$, the exponential map gives endpoint of the geodesic, namely
\begin{equation*}
    \begin{aligned}
        \Exp \left( \bm{P}_0, \dot{\bm{P}}(0) \right) & = \bm{P}_1 \\
			& =  \bm{P}_0^{\frac12} \exp \left(  \bm{P}_0^{\frac12} \Log \left( \bm{P}_0^{-\frac12} \bm{P}_1 \bm{P}_0^{-\frac12} \right)\bm{P}_0^{\frac12} \right) \bm{P}_0^{\frac12} \\
		    & = \bm{P}_0^{\frac12} \exp \left( \bm{P}_0^{-\frac12} \dot{\bm{P}}(0) \bm{P}_0^{-\frac12}\right) \bm{P}_0^{\frac12},\\
    \end{aligned}
\end{equation*}
where $\dot{\bm{P}}(0)$ is the tangent vector at the point $\bm{P}_0$.
Using the Riemannian distance \eqref{eq:Riedis}, we recall the RD mean and median of $m$ HPD matrices $\{ \bm{R}_i\}_{i=1}^m$ as follows
\begin{eqnarray}
		\overline{ \bm{R} }_{\mathrm{RDmean}}&: = \underset{ \bm{R} \in \mathscr{P}(N,\mathbb{C} )} \argmin \dfrac{1}{m}\sum\limits_{i=1}^{m} d_R^2(\bm{R}, \bm{R}_i), \\
		\overline{ \bm{R} }_{\mathrm{RDmedian}}&: = \underset{ \bm{R} \in \mathscr{P}(N,\mathbb{C} )} \argmin \dfrac{1}{m}\sum\limits_{i=1}^{m} d_R(\bm{R}, \bm{R}_i). 
\end{eqnarray}   

In other words, the RD mean and median are, respectively, solutions of the optimization problems with objective functions 
\begin{eqnarray}
        G_{\mathrm{RDmean}}\left( \bm{R} \right) = \frac{1}{m}\sum\limits_{i=1}^{m} d_R^2(\bm{R}, \bm{R}_i),\\ G_{\mathrm{RDmedian}}\left( \bm{R} \right) = \frac{1}{m}\sum\limits_{i=1}^{m} d_R(\bm{R}, \bm{R}_i).
\end{eqnarray}

\begin{prop}\label{prop1}
Since the RD mean and median are related to optimisation problems on HPD spaces as Riemannian manifolds, it is natural to propose the Riemannian gradient descent algorithms, i.e., gradient descent algorithms on Riemannian manifolds \cite{6556702,Robuststatistics}: 
\begin{equation}
    \begin{aligned}
		\bm{R}_{t+1} & = \Exp \left( \bm{R}_{t}, - \eta_t \operatorname{grad} G_{\mathrm{RDmean}} \left( \bm{R}_{t} \right) \right) \\
		& = \bm{R}_{t}^{\frac12} \exp \left( - \eta_t \bm{R}_{t}^{-\frac12} \operatorname{grad} G_{\mathrm{RDmean}} \left( \bm{R}_{t} \right) \bm{R}_{t}^{-\frac12} \right) \bm{R}_{t}^{\frac12} \\
        & = \bm{R}_{t}^{\frac12} \exp \left( - \frac{2 \eta_t}{m} \sum_{i=1}^{m} \Log \left( \bm{R}_{t}^{\frac12} \bm{R}_{i}^{-1} \bm{R}_{t}^{\frac12}\right) \right) \bm{R}_{t}^{\frac12},
    \end{aligned}
\end{equation}
and
\begin{equation}
    \begin{aligned}
		\bm{R}_{t+1} & = \Exp \left( \bm{R}_{t}, - \eta_t \operatorname{grad} G_{\mathrm{RDmedian}} \left( \bm{R}_{t} \right) \right) \\
		& = \bm{R}_{t}^{\frac12} \exp \left( - \eta_t \bm{R}_{t}^{-\frac12} \operatorname{grad} G_{\mathrm{RDmedian}} \left( \bm{R}_{t} \right) \bm{R}_{t}^{-\frac12}\right) \bm{R}_{t}^{\frac12} \\
	    & =\bm{R}_{t}^{\frac12} \exp \left(- \frac{\eta_t}{m} \sum_{i=1}^{m} \frac{ \Log \left( \bm{R}_{t}^{\frac12} \bm{R}_{i}^{-1} \bm{R}_{t}^{\frac12}\right) }{ d_R\left(\bm{R}_t ,\bm{R}_i \right) } \right) \bm{R}_{t}^{\frac12}.
    \end{aligned}
\end{equation}
Here  $\eta_t$ is the step size, and $\operatorname{grad}$ denotes the Riemannian gradient on Riemannian manifolds.  
\end{prop}

For more details on Riemannian gradient descent algorithms, the reader may refer to \cite{Smith1993,Udr1994}.  The RD mean and median can also be calculated numerically through the fixed-point algorithms \cite{fixpointHua}. It was noted that the detection performance of the RD mean in the matrix-CFAR is better than that of the RD median  \cite{fixpointHua}.

\subsection{The total Bregman divergence, TBD means and medians}
The Bregman divergence for matrices was introduced \cite{dhillon2008matrix} and this idea has been extended to total Bregman divergences \cite{myarticle}.
The total Bregman divergence of two matrices $\bm{A},\bm{B} \in \operatorname{GL}\left( N, \mathbb{C} \right)$, the general linear group of $N \times N$ invertible complex-valued matrices, is defined by 
\begin{equation}
	\hspace{0mm}\delta_F\left(\bm{A}, \bm{B} \right) =  \frac{F(\bm{A}) - F(\bm{B}) - \left\langle \nabla F(\bm{B}), \bm{A} - \bm{B} \right\rangle} {\sqrt{1 + \left\| \nabla F(\bm{B}) \right\|^2}}, 
\end{equation}
where $F \left(\bm{A} \right)$ is a differentiable and strictly convex function in $\operatorname{GL}\left( N, \mathbb{C} \right)$ and $\nabla$ denotes the (Euclidean) gradient with respect to the Frobenius metric, defined by 
\begin{equation}\label{eq:gra}
    \langle \nabla F(\bm{A}), \bm{B}\rangle = \frechetderi F(\bm{A}+\varepsilon \bm{B}), \ \ \bm{A},\bm{B}\in \operatorname{GL}(N,\mathbb{C}).
\end{equation}
The well-known  total square loss (TSL), and the total log-determinant  (TLD) and total von-Neumann  (TVN) divergences defined in the HPD manifold $ \mathscr{P}(N,\mathbb{C})$ are shown in Table \ref{table:TBD}.

\begin{table}[htb]
\caption{The TBDs} \label{table:TBD}
\centering
\begin{tabular}{c|c|c}
\hline
    & The function $F(\bm{Y})$  &  TBDs in $\mathscr{P}(N,\mathbb{C})$ \\ 
\hline 
TSL &  $\frac{1}{2} \left\| \bm{Y} \right\|^2$
    & $\delta_{\mathrm{TSL}}\left(\bm{Y}, \bm{Z} \right) = \cfrac{\left\| \bm{Y} - \bm{Z} \right\|^2}{2\sqrt{1 + \left\| \bm{Z} \right\|^2 }}$ \rule[0pt]{0pt}{20pt} \\
TLD divergence  &  $ - \ln \det \bm{Y} $
    & $\delta_{\mathrm{TLD}}\left(\bm{Y}, \bm{Z} \right) = \cfrac{\ln \det \left(\bm{Z} \bm{Y}^{-1}\right) + \tr \left( \bm{Z}^{-1} \bm{Y} \right) - N}{2\sqrt{1 + \left\| \bm{Z}\hermconj \right\|^2 }}$ \rule[0pt]{0pt}{20pt} \\
TVN divergence  &  $\tr \left(  \bm{Y} \Log \bm{Y} - \bm{Y}\right)$
    &  $\delta_{\mathrm{TVN}}\left(\bm{Y}, \bm{Z} \right) = \cfrac{\tr \left(  \bm{Y} \left( \Log \bm{Y} - \Log \bm{Z}\right) - \bm{Y} + \bm{Z} \right)}{2\sqrt{1 + \left\| \left( \Log \bm{Z} \right)\hermconj \right\|^2 }}$ \rule[0pt]{0pt}{20pt} \\
\hline
\end{tabular}
\end{table}

%
	
Before introducing the TBD medians, we briefly recall their mean counterparts. Details can be found in \cite{myarticle}. 
The TBD mean of $m$ HPD matrices $\{ \bm{R}_i\}_{i=1}^m$ is defined by 
\begin{eqnarray}\label{eq:TBDMeanDef}
	\underset{ \bm{R} \in \mathscr{P}(N,\mathbb{C} )} \argmin \dfrac{1}{m}\sum_{i=1}^{m} \delta_F\left(\bm{R},\bm{R}_i \right). 
\end{eqnarray}
As the TBD $\delta_F(\bm{R},\bm{R}_i)$ is defined in the convex HPD space and is strictly convex with respect to $\bm{R}$, the TBD mean exists uniquely. 
It solves the next matrix equation,
\begin{eqnarray}\label{eq:TBDMeanSolution}
	\nabla F(\bm{R}) = \left({\sum\limits_{j=1}^{m} \frac{1}{\sqrt{1+\norm{\nabla F(\bm{R}_j)}^2}}}\right)^{-1}\left({\sum_{i=1}^{m} \frac{\nabla F(\bm{R}_i)}{\sqrt{1+\norm{\nabla F(\bm{R}_i)}^2}}}\right).
\end{eqnarray}
    
In particular, the TBD means corresponding to the convex functions $F$ given by Table \ref{table:TBD}  can be calculated analytically \cite{myarticle}. 
We summarize them up as follows. 
The TSL mean, TLD mean, and TVN mean of  $m$ HPD matrices $\{ \bm{R}_i \}_{i=1}^m$ are respectively given by
\begin{eqnarray}
    \overline{\bm{R}}_{\mathrm{TSL}} &=& \left( \sum\limits_{j=1}^m \frac{1}{\sqrt{1+\norm{\bm{R}_j}^2}} \right)^{-1} \left(\sum_{i=1}^{m} \frac{\bm{R}_i}{\sqrt{1+\norm{ \bm{R}_i }^2}}\right), \label{eq:TSLmeanfixpoint} \\
    \overline{\bm{R}}_{\mathrm{TLD}} &=& \left( \sum_{j=1}^m \frac{1}{\sqrt{1+\norm{\bm{R}^{-1}_j}^2}} \right) \left(\sum_{i=1}^{m}\frac{\bm{R}_i^{-1}}{\sqrt{1+\norm{\bm{R}_i^{-1}}^2} }\right)^{-1} ,\label{eq:TLDmeanfixpoint} \\
    \overline{\bm{R}}_{\mathrm{TVN}} &=& \exp\left(  \left( \sum_{j=1}^m \frac{1}{\sqrt{1+\norm{\Log \bm{R}_j}^2}} \right)^{-1}  \sum_{i=1}^{m} \frac{ \Log \bm{R}_i}{\sqrt{1+\norm{\Log \bm{R}_j}^2}}\right). \label{eq:TVNmeanfixpoint} 
\end{eqnarray}

\begin{defn}
Assume $F$ is a differentiable and strictly convex function. Let   $\delta_F$ be the associated TBD. The TBD median of $m$ HPD matrices $\{ \bm{R}_i\}_{i=1}^m$ is defined by
	\begin{equation}\label{eq:TBDMedianDef}
		\underset{ \bm{R} \in \mathscr{P}(N,\mathbb{C} )} \argmin \dfrac{1}{m}\sum_{i=1}^{m} \left\{\delta_F\left(\bm{R},\bm{R}_i \right) \right\}^{\frac{1}{2}}.
	\end{equation}
\end{defn}

\begin{thm}
    \label{prop:TBDmedian}
	If the TBD median \eqref{eq:TBDMedianDef} exists, then it satisfies  the matrix equation,
	\begin{equation}\label{eq:TBDMedianSolution}
	    \nabla F\left( \bm{R}\right)=\left( \sum_{j=1}^m \frac{1}{\left\{\delta_F\left(\bm{R},\bm{R}_j \right) \right\}^{\frac{1}{2}} \sqrt{ 1 + \left\| \nabla F(\bm{R}_j) \right\|^2} } \right)^{-1} \left(\sum_{i=1}^m \frac{\nabla F\left( \bm{R}_i\right)}{\left\{\delta_F\left(\bm{R},\bm{R}_i \right) \right\}^{\frac{1}{2}} \sqrt{ 1 + \left\| \nabla F(\bm{R}_i) \right\|^2} }\right).  
	\end{equation}
\end{thm}

\begin{proof}
Let  $G(\bm{R})$ be the function to be optimized, namely 
    \begin{equation}
        G(\bm{R}) := \frac{1}{m} \sum\limits_{i = 1}^{m} \left\{\delta_F\left( \bm{R},\bm{R}_i \right) \right\}^\frac{1}{2}.
    \end{equation}
The proof is completed by setting $\nabla G(\bm{R}) = 0$ where the gradient of $G(\bm{R})$ is
    \begin{equation}
            \nabla G(\bm{R}) = \frac{1}{2m} \sum_{i=1}^{m} \frac{\nabla F(\bm{R}) - \nabla F(\bm{R}_i)}{ \left\{\delta_F\left(\bm{R},\bm{R}_i \right) \right\}^\frac{1}{2}  \sqrt{1+\norm{\nabla F(\bm{R}_i)}^2}}. 
    \end{equation}
\end{proof} 
Unfortunately, it is difficult to solve the matrix equation \eqref{eq:TBDMedianSolution} analytically, and alternatively we seek for its numerical solutions using the fixed-point algorithm \cite{MoaAvg}.

\begin{prop}
The TSL median, the TLD median, and the TVN median of $m$ HPD matrices $\{ \bm{R}_i\}_{i=1}^m$ can respectively be calculated by the following fixed-point algorithms in Table \ref{table:fixed-point algorithms}.
\end{prop}

\begin{table}[htb]
\caption{The fixed-point algorithms for computing TBD medians} \label{table:fixed-point algorithms}
\centering
\small 
\begin{tabular}{c|c}
\hline 
TBD & The fixed-point algorithm\\ 
\hline 
TSL  & 
    $ \overline{\bm{R}}_{t+1} =  \left( \sum\limits_{j=1}^{m} \cfrac{1}{\left\{\delta_{\mathrm{TSL}}\left(\overline{\bm{R}}_{t},\bm{R}_j \right) \right\}^\frac{1}{2} \sqrt{1+\norm{\bm{R}_j}^2}}\right)^{-1}\left(\sum\limits_{i=1}^{m} \cfrac{ \bm{R}_i }{ \left\{\delta_{\mathrm{TSL}}\left(\overline{\bm{R}}_{t},\bm{R}_i \right) \right\}^\frac{1}{2} \sqrt{1+\norm{\bm{R}_i}^2}}\right)  $\\
TLD  & 
    $  \overline{\bm{R}}_{t+1} = \left(\sum\limits_{j=1}^{m} \cfrac{1}{\left\{\delta_{\mathrm{TLD}}\left(\overline{\bm{R}}_{t},\bm{R}_j \right) \right\}^\frac{1}{2} \sqrt{1+\norm{\bm{R}_j^{-1}}^2}}\right)  \left( {\sum\limits_{i=1}^{m} \cfrac{ \bm{R}_i^{-1} }{ \left\{\delta_{\mathrm{TLD}}\left(\overline{\bm{R}}_{t},\bm{R}_i \right) \right\}^\frac{1}{2} \sqrt{1+\norm{\bm{R}_i^{-1}}^2} }} \right)^{-1} $\\
TVN  & 
    $ \overline{\bm{R}}_{t+1} = \mathrm{exp} \left( \left( \sum\limits_{j=1}^{m} \cfrac{1}{\left\{\delta_{\mathrm{TVN}}\left(\overline{\bm{R}}_{t},\bm{R}_j \right) \right\}^\frac{1}{2} \sqrt{1+\norm{\left( \Log \bm{R}_j \right)\hermconj }^2}}\right)^{-1}    
    \sum\limits_{i=1}^{m} \cfrac{ \Log \bm{R}_i }{ \left\{\delta_{\mathrm{TVN}}\left(\overline{\bm{R}}_{t},\bm{R}_i \right) \right\}^\frac{1}{2} \sqrt{1+\norm{ \left( \Log \bm{R}_i \right)\hermconj}^2} } \right)$\\
\hline
\end{tabular}
\end{table}


In the matrix-CFAR, the matrix $\bm{R}_g$ in \eqref{eq:decison_of_R_g} will be replaced by the TBD means, the TBD medians, the RD mean, and so forth; see also \Figref{fig:processMCFAR}.

Numerical simulations in the next sections show that the Riemannian gradient descent algorithms introduced in Proposition \ref{prop1} and the fixed-point algorithms in Table \ref{table:fixed-point algorithms} are all convergent with the initial value chosen as the arithmetic mean of a set of HPD matrices; the tolerance of each iteration is $10^{-3}$.

\section{Numerical performance analysis}\label{sec:DP}
To confirm and compare the performance of the matrix-CFAR via TBD means and medians, the following numerical simulations are conducted. 
For the comparison, we also show results of the RD mean induced from the AIRM, the GLRT, and the ANMF.
Two types of clutter are considered. 
One is the Gaussian clutter modeled by an $N$ dimensional random complex vector possessing the complex circular Gaussian distribution $ CN \left( 0, \bm{\Sigma} \right)$ of zero-mean and a covariance matrix $\bm{\Sigma}$. Its probability density function is given by
\begin{equation}
	p\left(\bm{z}\ |\ 0, \bm{\Sigma}\right) = \frac{1}{\pi^N \det\left(\bm{\Sigma}\right)} \exp \left(- \bm{z}\hermconj \bm{\Sigma}^{-1} \bm{z} \right).
\end{equation}
The other is the compound-Gaussian clutter which is a model commonly used for characterising heavy-tailed clutter distributions, for instance sea clutter \cite{seaclutter}, defined by 
\begin{equation}
	\bm{c} = \sqrt{\tau} \bm{z}. 
\end{equation}
Here $\bm{z}$ is called the speckle component which is fast fluctuating and modeled by a complex Gaussian distribution.
The positive real random variable $\tau$ is called the texture component which is comparatively slowly varying. 
In the simulation, the texture component $\tau$ is assumed to subject to a gamma distribution with a shape parameter $ \alpha$ and a scale parameter $\beta$ whose probability density function is 
\begin{equation}
	q\left(\tau\ |\ \alpha, \beta\right) = \frac{1}{\beta^\alpha \Gamma\left(\alpha\right)} \tau^{\alpha -1} \exp \left(\frac{- \tau}{\beta} \right),
\end{equation}
where $\Gamma (\cdot)$ denotes the gamma function. Consequently, the distribution of clutter is defined by the complex K-distribution which is a special compound-Gaussian model \cite{seaclutter}. 
In the simulations, we generate the Gaussian clutter, as well as the speckle component in the compound-Gaussian clutter, from $CN\left(0,\bm{\Sigma} \right)$ of which the known covariance matrix $\bm{\Sigma}$ is given by 
\begin{equation}\label{eq:sig}
	\bm{\Sigma} = \bm{\Sigma}_0 + \bm{I}_N,
\end{equation}
where 
\begin{equation}
	\bm{\Sigma}_0 \left(i,j\right) = \sigma_c^2 \rho^{\left| i-j \right|} \exp \left( \mathrm{i} 2\pi f_c (i-j) \right), \nonumber \ \ i,j = 1,2, \ldots , N.
\end{equation}
Here, $\bm{I}_N$ is the $N$ dimensional identity matrix. In the simulation,  the clutter-to-noise ratio $\sigma_c$ is set to $\sigma_c^2=20 \mathrm{dB}$, the one-lag coefficient is $\rho =0.9$, and  the normalized Doppler frequency of  clutter is $f_c = 0.2$. Furthermore, the shape parameter is set to $\alpha = 4 $ and the scale parameter is set to $\beta = 3$. The signal to clutter ratio (SCR) is defined by 
\begin{equation}\label{eq:SCR}
    \mathrm{SCR} = |\xi|^2 \bm{s}\hermconj \bm{R}^{-1} \bm{s},
\end{equation}
where $\xi$ is the amplitude coefficient, $\bm{s}$ is the target vector defined as  \eqref{eq:target} and $\bm{R}$ is the clutter covariance matrix.
Furthermore, the false alarm rate $P_{fa}$ is $10^{-3}$, and the signal normalized Doppler frequency $f_d$ is $0.2$.
Two interferences are inserted into clutter data and its normalized Doppler frequency $f_i$ is also $0.2$.
Dimension of the signal is set to $N=8$, and we study the cases when the numbers of observations are $m=8,12,16$, respectively. 

The threshold $\gamma$ is derived by $100/P_{fa}$, and the probability of detection $P_d$ is derived by 2000 independent trails. 
\Figref{fig:detectionperformance} shows that in both types of clutter, the detection performance improves as the number of observations increases. The GLRT and the ANMF are not valid when the number of observations is small but the matrix-CFAR detectors work well.
In the case of Gaussian clutter (see Figs. \ref{fig:m=8gaussian}, \ref{fig:m=12gaussian}, \ref{fig:m=16gaussian}), the TBD means and medians are advantageous over the RD mean and median, the GLRT and the ANMF and in particular, the TLD median has the best performance. In the compound-Gaussian clutter (see Figs. \ref{fig:m=8gamma}, \ref{fig:m=12gamma}, \ref{fig:m=16gamma}), the TVN median has the best performance when the number of observations is not large enough. In the case of $m=16$, the ANMF is the best among all detectors.
Among the TBD means and medians, the TLD and TVN medians behave better than the corresponding  means, but surprisingly the TSL median is worse than the TSL mean.

\begin{rem}
Generally speaking, the simulations show that the matrix-CFAR  performs  better than the GLRT and ANMF particularly when limited data is available.  This is probably due to that the effect of clutter and noise in the matrix-CFAR is less than that in the GLRT and ANMF, since
the threshold of the GLRT and ANMF is derived by using contaminated sample data directly, while the matrix-CFAR leverages the covariance matrix of sample data. 
\end{rem}

\begin{figure}[htbp]
    \begin{minipage}{0.48\hsize}
        \centering
        \includegraphics[width=\hsize]{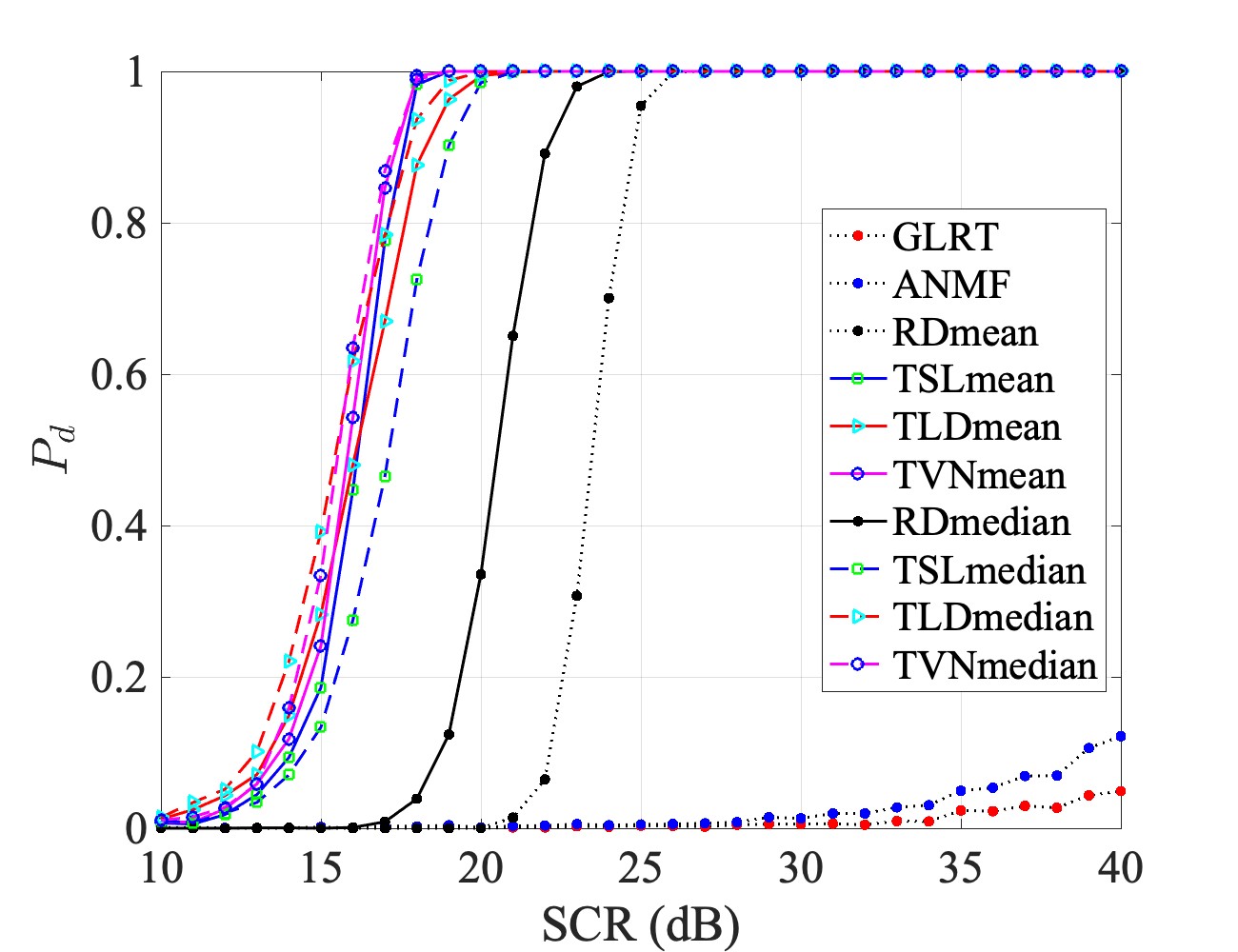}
        \subcaption{Gaussian clutter, $m=8$}\label{fig:m=8gaussian}
    \end{minipage}
    \begin{minipage}{0.48\hsize}
        \centering
        \includegraphics[width=\hsize]{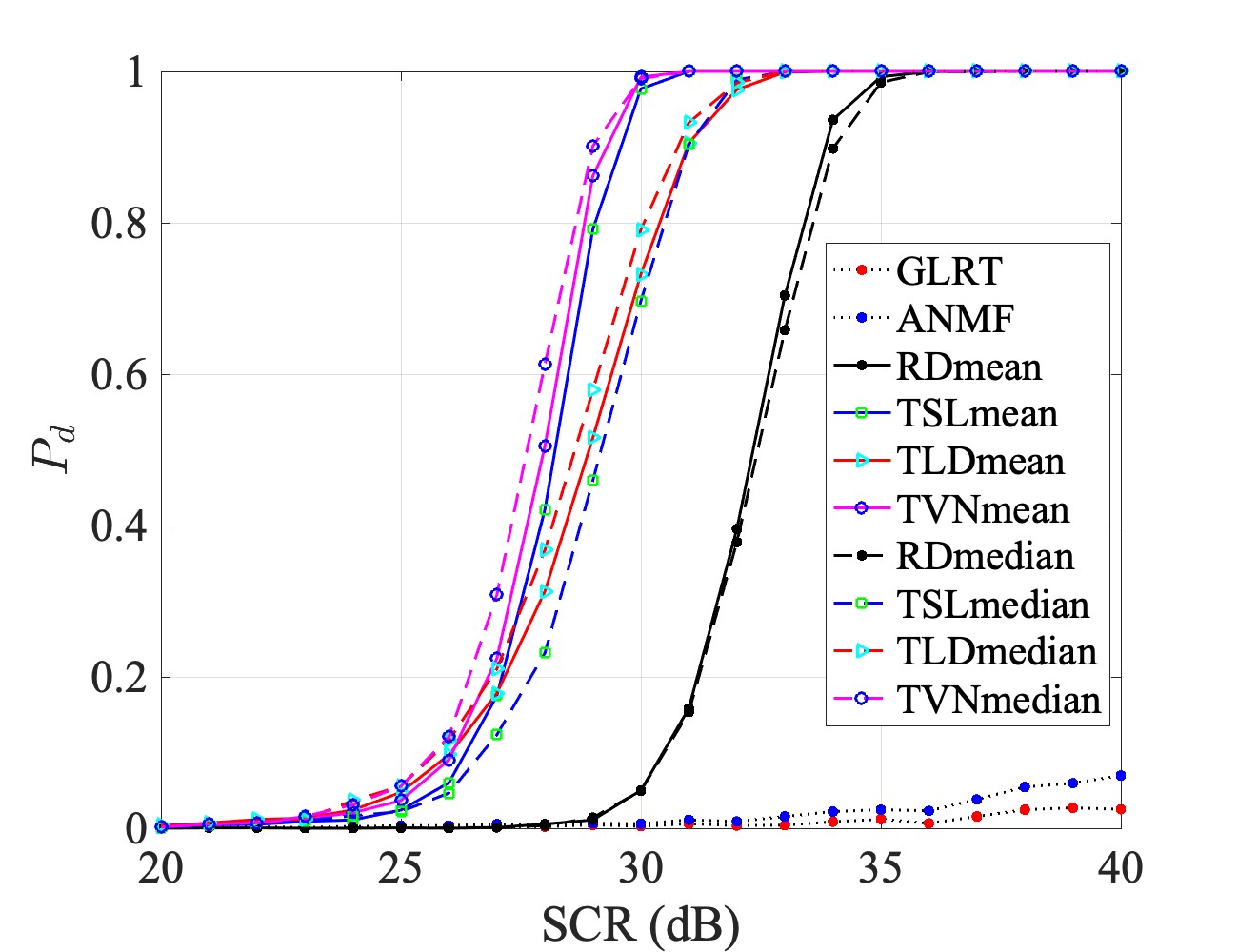}
        \subcaption{Compound-Gaussian clutter, $m=8$}\label{fig:m=8gamma}
    \end{minipage}
    \begin{minipage}{0.48\hsize}
        \centering
        \includegraphics[width=\hsize]{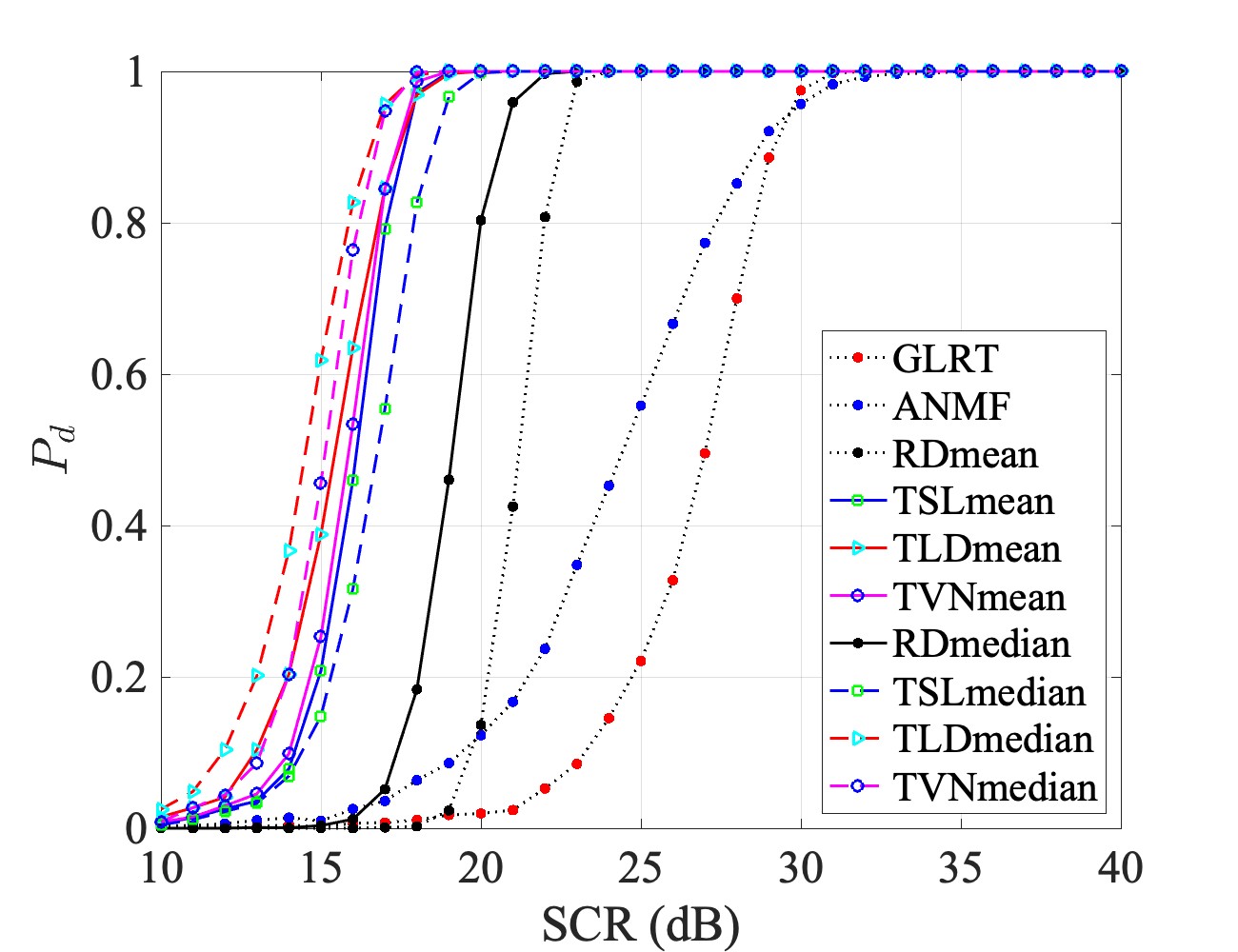}
        \subcaption{Gaussian clutter, $m=12$}\label{fig:m=12gaussian}
    \end{minipage}
    \begin{minipage}{0.48\hsize}
        \centering
        \includegraphics[width=\hsize]{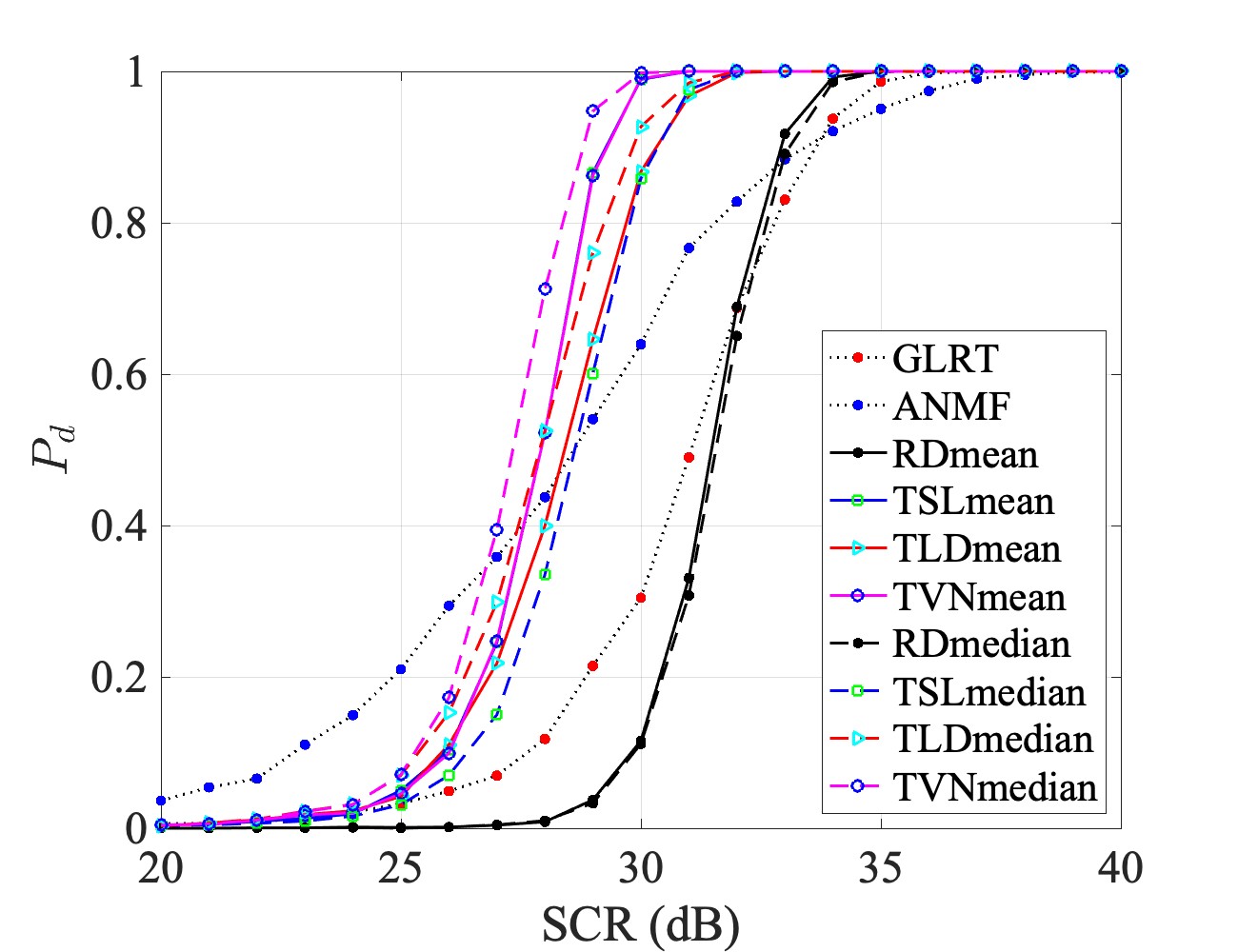}
        \subcaption{Compound-Gaussian clutter, $m=12$}\label{fig:m=12gamma}
    \end{minipage}
    \begin{minipage}{0.48\hsize}
        \centering
        \includegraphics[width=\hsize]{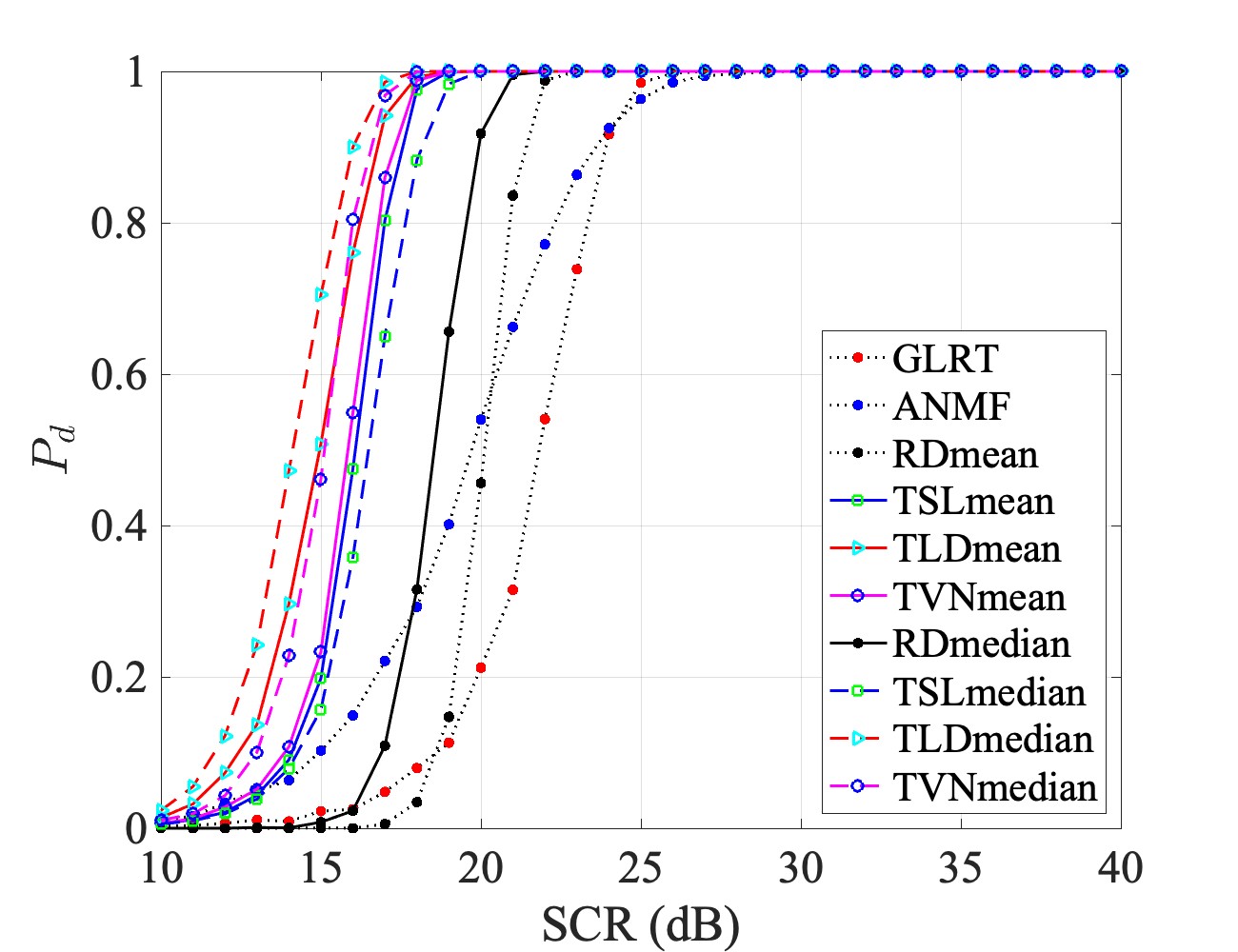}
        \subcaption{Gaussian clutter, $m=16$}\label{fig:m=16gaussian}
    \end{minipage}
    \begin{minipage}{0.48\hsize}
        \centering
        \includegraphics[width=\hsize]{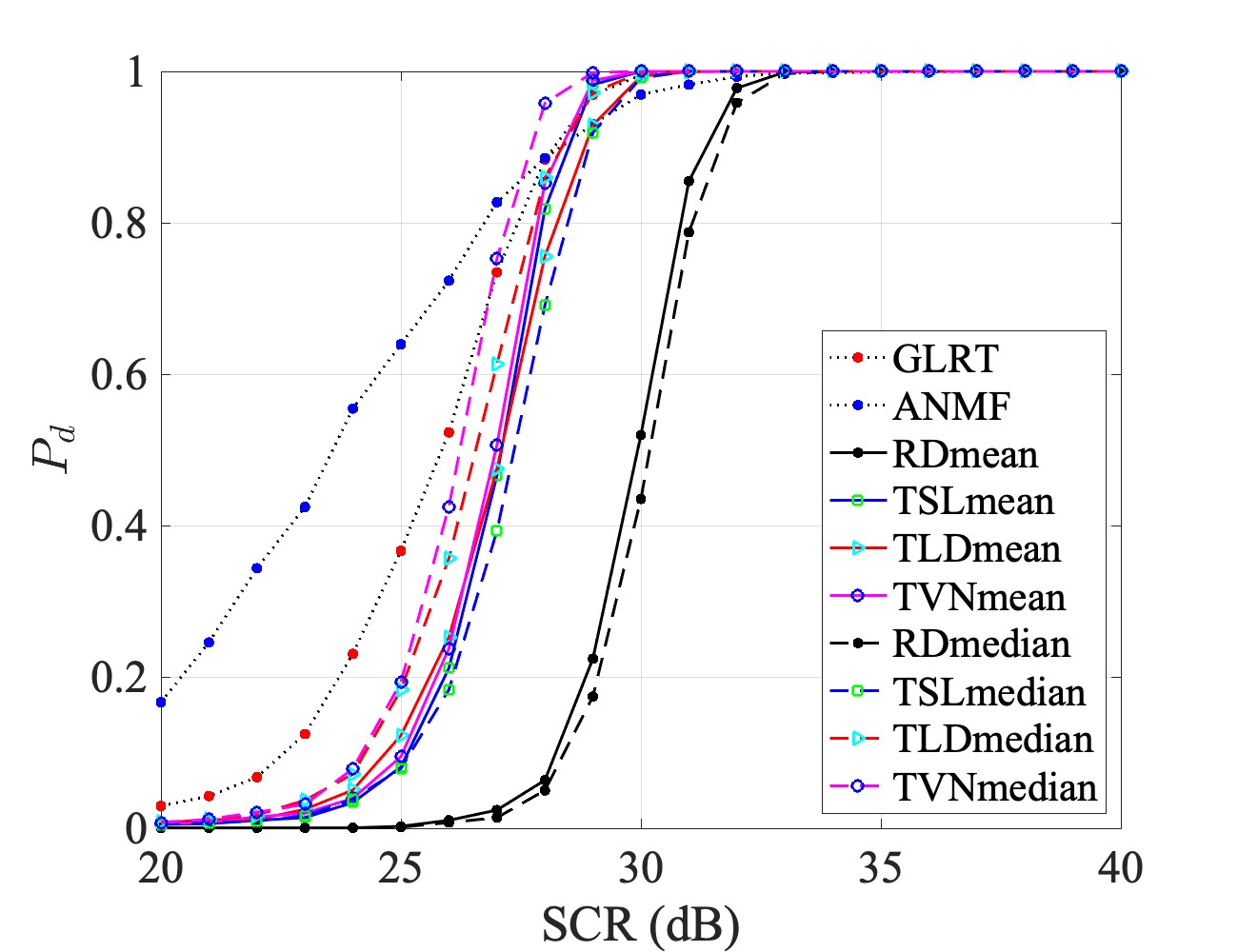}
        \subcaption{Compound-Gaussian clutter, $m=16$}\label{fig:m=16gamma}
    \end{minipage}
\caption{$P_d$ versus signal-to-clutter ratio for the TBD means and medians, the RD mean and median, the GLRT and the ANMF.}\label{fig:detectionperformance}
\end{figure}

\section{Robustness analysis}
\label{sec:RA}
In real detection, the received signal may also include outliers, and hence robustness about estimate of the covariance matrix is crucial. To evaluate the robustness of the detectors to outliers, the influence function has usually been employed \cite{myarticle}. It was noted that the influence functions with respect to  the TSL mean and the TLD mean can easily be calculated analytically but not for the TVN mean \cite{myarticle}. The situation for the median counterparts is often more complicated. In this section, we will propose an orthogonal basis for Hermitian matrices; by doing so, we are allowed to rewrite the equations of influence functions as a linear system, that can be solved analytically in a routine manner. 


\subsection{An orthogonal basis for Hermitian matrices}
As we will see later that influences functions are defined as a small perturbation of the means or medians that depends on a Hermitian matrix.  A well-known basis for the $N^2$ dimensional space of Hermitian matrices  $\mathscr{H}(N,\mathbb{C})$ is given by $N^2$ number of matrices:
\begin{equation}
	\begin{aligned}
		&\bm{E}_{ii}, &1\leq i\leq N; \\
		&\bm{E}_{ij} + \bm{E}_{ji},  ~& 1\leq i<j\leq N; \\
		&\mathrm{i} \left( \bm{E}_{ij} - \bm{E}_{ji} \right),  ~ &1\leq i<j\leq N, 
	\end{aligned}
\end{equation} 
where $\bm{E}_{ij}$ represents the $N \times N$ matrix with 1 at row $i$ and column $j$ and 0 elsewhere. 
To ease the application in the following robustness analysis, we normalize them into the following orthogonal  basis
%
\begin{equation}\label{eq:bas}
	\begin{aligned}
		&\bm{E}_{ii}, &1\leq i\leq N; \\
		&\frac{1}{\sqrt{2}} \left( \bm{E}_{ij} + \bm{E}_{ji}\right) , ~ &1\leq i<j\leq N ; \\
	 	&\frac{\mathrm{i}}{\sqrt{2}}\left( \bm{E}_{ij} - \bm{E}_{ji} \right), ~  &1\leq i<j\leq N,
	\end{aligned}
\end{equation} 
and denote it by $\{\mathsf{E}_{k}\}_{k=1}^{N^2}$. 
Each Hermitian matrix $\bm{H} \in \mathscr{H}(N,\mathbb{C})$ can be represented  as a linear combination using  the orthogonal basis, namely
\begin{equation}
	\bm{H} = h^{k} \mathsf{E}_{k}, \quad h^k\in\mathbb{R},
\end{equation}
where the summation is over $\{1,\ldots,N^2\}$ and the Einstein summation convention is applied here and all through the paper.
It is immediate to check that the basis $\{\mathsf{E}_{k}\}_{k=1}^{N^2}$ is orthogonal by showing that 
\begin{equation}
	\langle \mathsf{E}_{k} ,\mathsf{E}_{l} \rangle = \delta_{k l},\quad \forall  k,l = 1,\ldots, N^2,
\end{equation}
where $\delta_{kl}$ is the Kronecker delta that gives 1 when $k=l$ and 0 otherwise. Given a  Hermitian matrix  $\bm{H}$, it is immediate to write down its coordinate expression as
\begin{equation}
\bm{H}=h^k\mathsf{E}_k, \text{ with }  h^k=\langle\bm{H},\mathsf{E}_k\rangle.
\end{equation}

\subsection{Robustness via influence functions}
Let $\overline{\bm{R}}$ be the TBD medians (or TBD means, RD mean, RD median) of a given set of $m$ HPD matrices $\{ \bm{R}_i\}_{i=1}^m$, and let $\widehat{\bm{R}}$ be the TBD medians (or TBD means, RD mean, RD median) of the mixed HPD data including $n$ number of outliers $\{ \bm{P}_j\}_{j=1}^n$. 
By thinking of the outliers as a perturbation $\varepsilon \left( \varepsilon << 1 \right)$ of each mean or median, $\widehat{\bm{R}}$ and $\overline{\bm{R}}$ can be related by 
\begin{equation}
	\widehat{\bm{R}} = \overline{\bm{R}} + \varepsilon \bm{H} \left( \{ \bm{R}_i\}_{i=1}^m,\{ \bm{P}_j\}_{j=1}^n \right) + \mathcal{O} \left( \varepsilon^2\right),
\end{equation}
where $ \bm{H} \left( \{ \bm{R}_i\}_{i=1}^m,\{ \bm{P}_j\}_{j=1}^n \right) $ is a Hermitian matrix depending on the  outliers, and the $m$ HPD matrices $\{ \bm{R}_i\}_{i=1}^m$, the latter of which are assumed to be given and fixed.
The function 
\begin{equation}\label{eq:inff}
 f\left( \{ \bm{R}_i\}_{i=1}^m,\{ \bm{P}_j\}_{j=1}^n \right) \coloneqq\dfrac{\norm{\bm{H}\left( \{ \bm{R}_i\}_{i=1}^m,\{ \bm{P}_j\}_{j=1}^n \right)}}{\norm{\overline{\bm{R}}}} 
\end{equation}
is defined as the (normalized)  {\it influence function}.

To calculate the corresponding influence functions, define the objective function $G(\bm{R})$  for the contaminated $m+n$ number of HPD matrices as follows 
\begin{equation}
	G(\bm{R}):= (1-\varepsilon)\frac{1}{m}\sum_{i=1}^{m} d\left(\bm{R},\bm{R}_i \right) + \varepsilon\frac{1}{n}\sum_{j=1}^{n} d\left(\bm{R},\bm{P}_j \right),
\end{equation}
where $d$ is either the distance or divergence functions for medians, and square of them for means.
Since $\bm{\widehat{R}}$ is any mean or median of the mixed $m+n$ HPD matrices, we have $\nabla G(\bm{\widehat{R}}) = \bm{0},$\footnote{Note that in the Riemannian case, this should be $\operatorname{grad} G(\widehat{\bm{R}})=0$, which is nevertheless equivalent to $\nabla G(\widehat{\bm{R}})=0$ for HPD manifolds equipped with the AIRM.} that is, 
\begin{equation}
		(1-\varepsilon) \frac{1}{m}\sum_{i=1}^{m} \nabla d\left(\bm{\widehat{R}},\bm{R}_i \right) +  \varepsilon\frac{1}{n}\sum_{j=1}^{n} \nabla d\left(\bm{\widehat{R}},\bm{P}_j \right) = \bm{0}.
\end{equation}
Differentiating $\nabla G(\widehat{\bm{R}})=\bm{0}$ about$\varepsilon$ at $\varepsilon=0$ , we obtain a matrix equation for $\bm{H}$:
\begin{equation}\label{eq:inf0}
	\hspace{-8mm}\frac{1}{m}\frechetderi \left( \sum_{i=1}^{m} \nabla d\left(\bm{\widehat{R}},\bm{R}_i \right) \right) + \frac{1}{n}\sum_{j=1}^{n} \nabla d\left(\overline{\bm{R}},\bm{P}_j \right) = \bm{0} .
\end{equation}
Next, we are going to show that the above equation can be solved analytically using the orthogonal basis \eqref{eq:bas} defined above. 

Recall that the gradient of a function defined in HPD spaces is a Hermitian matrix. Using the orthogonal basis $\{\mathsf{E}_{k}\}_{k=1}^{N^2}$, we can write
\begin{equation}\label{eq:J}
\sum_{i=1}^{m} \nabla d\left(\bm{\widehat{R}},\bm{R}_i \right)=\xi^l(\widehat{\bm{R}})\mathsf{E}_l
\end{equation}
and
\begin{equation}
\sum_{j=1}^{m} \nabla d\left(\overline{\bm{R}},\bm{P}_j \right) =\phi^l\mathsf{E}_l.
\end{equation}
Letting $\bm{H}=h^s\mathsf{E}_s$, the first term in \eqref{eq:inf0} then becomes 
\begin{eqnarray}\label{eq:inf1}
 \dfrac{1}{m}\frechetderi \xi^l(\widehat{\bm{R}})\mathsf{E}_l &=&\dfrac{1}{m}\frechetderi \xi^l(\overline{\bm{R}}+\varepsilon\bm{H})\mathsf{E}_l \nonumber \\
&=&\dfrac{1}{m}\langle \nabla \xi^l(\overline{\bm{R}}),\bm{H}\rangle \mathsf{E}_l \nonumber \\
&=&\dfrac{h^k}{m}\langle \nabla \xi^l(\overline{\bm{R}}),\mathsf{E}_k\rangle \mathsf{E}_l.
\end{eqnarray}
Definition of the  gradient \eqref{eq:gra} (induced to HPD spaces) is applied here. Although a linear system about $\{h^k\}_{k=1}^{N^2}$  can be obtained, unfortunately, as $\{\xi^l(\overline{\bm{R}})\}_{l=1}^{N^2}$ is the coordinates whose gradient can be practically difficult to be calculated, in particular, in numerical simulations. Alternatively, noticing the first term in \eqref{eq:inf0} is Hermitian and linear about $\{h^k\}_{k=1}^{N^2}$ (from \eqref{eq:inf1}), we can simply write 
\begin{equation}
\frechetderi \left( \sum_{i=1}^{m} \nabla d\left(\bm{\widehat{R}},\bm{R}_i \right) \right)=h^k\theta_k^l\mathsf{E}_l.
\end{equation}
Consequently, we have successfully written the matrix equation \eqref{eq:inf0} as a linear system about $\{h^k\}_{k=1}^{N^2}$ with matrix coefficients:
\begin{equation}
\frac{1}{m}h^k\theta_k^l\mathsf{E}_l+\frac{1}{n}\phi^l\mathsf{E}_l=\bm{0}.
\end{equation}
%
Conducting Frobenius metric on both sides with each $\mathsf{E}_s$ and making use of the orthogonality, we obtain a linear system for $\{h^k\}_{k=1}^{N^2}$, reading
\begin{equation}\label{eq:inf00}
\frac{1}{m}\theta_{k}^sh^k+\frac{1}{n}\phi^s=0, \quad s=1,\ldots,N^2,
\end{equation} 
which can be immediately solved if the coefficient matrix $(\theta_{k}^s)$ is invertible. We summarize the above deduction into the following theorem.

\begin{thm}
Under the orthogonal basis $\{\mathsf{E}_{k}\}_{k=1}^{N^2}$ defined by \eqref{eq:bas}, the Hermitian matrix $\bm{H}=h^k\mathsf{E}_k$ in the influence function \eqref{eq:inff} is given by linear system \eqref{eq:inf00}.
\end{thm}


{\bf Influence function of the RD mean.}  Taking the RD mean as an illustrative example, we show how the above deduction helps us to derive the matrix $\bm{H}$ and hence the corresponding influence function. 
Regarding the contaminated HPD matrices, define $G(\bm{R})$ as the objective function, i.e.,
\begin{equation}\label{eq:RDmeaninf}
    G(\bm{R}):= (1-\varepsilon)\frac{1}{m}\sum_{i=1}^{m} d^2_R\left(\bm{R},\bm{R}_i \right) + \varepsilon\frac{1}{n}\sum_{j=1}^{n} d^2_R\left(\bm{R},\bm{P}_j \right), 
\end{equation}
where 
\begin{equation}
	d^2_R\left(\bm{R},\bm{R}_i \right) =  \norm{\Log \left( \bm{R}_i^{-1} \bm{R} \right)}^2 = \tr \left( \Log^2 \left( \bm{R}_i^{-1} \bm{R} \right) \right) . \\
\end{equation}
Differentiate $ \operatorname{grad} G(\widehat{\bm{R}})=\bm{0}$ with respect to $\varepsilon$ at $\varepsilon=0$ and obtain 
\begin{equation}
        \hspace{-7mm}\frac{1}{m}\frechetderi \sum_{i=1}^{m} \bm{\widehat{R}} \Log\left(\bm{R}_i^{-1} \bm{\widehat{R}}\right) + \frac{1}{n}\sum_{j=1}^{n} \overline{\bm{R}} \Log \left(\bm{P}_j^{-1} \overline{\bm{R}} \right) = \bm{0}. 
\end{equation}
The second term can be simply written as 
\begin{equation}
 \frac{1}{n} \sum_{j=1}^{n} \overline{\bm{R}} \Log \left(\bm{P}_j^{-1} \overline{\bm{R}} \right) = \frac{1}{n}\phi^{k} \mathsf{E}_{k},
\end{equation}
while the first term can be 
expanded as \cite{myarticle,Moa2005}
\begin{equation*} 
\frac{1}{m} \sum_{i=1}^{m} \left( \bm{{H}} \Log \left(\bm{R}_i^{-1} \overline{\bm{R}}\right) + \overline{\bm{R}} \int_0^1 \left[ ( \bm{R}_i^{-1} \overline{\bm{R}} - \bm{I})\tau+\bm{I} \right]^{-1} \bm{R}_i^{-1} \bm{H} \left[ ( \bm{R}_i^{-1} \overline{\bm{R}} - \bm{I})\tau+\bm{I} \right]^{-1} d\tau \right).
\end{equation*}
Substituting $\bm{H}=h^k\mathsf{E}_k$ to the first term, it can be written in the form $\frac{1}{m}h^k\theta_k^s\mathsf{E}_s$, where $\theta_k^s$ is given by
\begin{equation*} 
\theta_k^s=\left\langle  \sum_{i=1}^{m} \left(\mathsf{E}_k\Log \left(\bm{R}_i^{-1} \overline{\bm{R}}\right)+ \overline{\bm{R}} \int_0^1 \left[ ( \bm{R}_i^{-1} \overline{\bm{R}} - \bm{I})\tau+\bm{I} \right]^{-1} \bm{R}_i^{-1} \mathsf{E}_k \left[ ( \bm{R}_i^{-1} \overline{\bm{R}} - \bm{I})\tau+\bm{I} \right]^{-1} d\tau \right),\mathsf{E}_s\right\rangle.
\end{equation*}

Finally, we obtain the linear system \eqref{eq:inf00} about $\{h^k\}_{k=1}^{N^2}$ by conducting the Frobenius metric with respect to each element of the orthogonal basis, which will give us the influence function for the RD mean. Those of the RD median, as well as  the TBD means and medians can also be analytically computed in a similar way; details are omitted here. 

In the following simulations, we firstly generate $m$ data of sample, each of which is  $N$ dimensional and complex valued, generated  from $CN \left( 0, \bm{\Sigma} \right)$ where the  covariance matrix  $\bm{\Sigma}$  is given by \eqref{eq:sig}. 
The means and medians $ \overline{\bm{R}}$ are calculated from the $m$ HPD covariance matrices $\{ \bm{R}_i\}_{i=1}^m$ of the $m$ number of sample data. We then mix $n$ outliers $\{ \bm{P}_j\}_{j=1}^n$  with the sample data; the outliers are  modeled as the covariance matrix of $\xi \bm{s} + \bm{c}$  where $\bm{s}$ denotes the steering vector, $\bm{c}$ is the Gaussian clutter with the covariance matrix $\bm{\Sigma}$,  and $\xi$ is the amplitude coefficient derived by the signal-to-clutter ratio which is set to 40 dB 
as \eqref{eq:SCR}.
In the numerical simulations, the number of sample is $m=50$ while the number of outliers $n$ varies from 1 to 40. 
Regarding the influence function of the SCR, the $\overline{\bm{R}}$ is calculated from the $m$ number of sample data and the $\widehat{\bm{R}}$ is calculated from the $m+n$ number contaminated data by using \eqref{eq:SCM}. Therefore the influence function of the SCR is derived similarly to that in \eqref{eq:inff}.
\Figref{fig:Influencefunction} plots the influence functions which are calculated by averaging 1000 trails and are shown in the logarithmic scale. It is noticed that the influence functions of the TBD medians are smaller than the TBD means and the geometric means and 
medians except for the TSL mean are more robust than the SCM. 
However, regardless of its worse detection performance compared with the TBD means and medians, the RD mean and median are more robust, although the RD mean and median take much longer time than the TBD means and medians. Therefore, from a comprehensive perspective the TBD medians can serve as effective estimators.

\begin{figure}[htbp]
	\centering
	\includegraphics[width = 0.9\linewidth ]{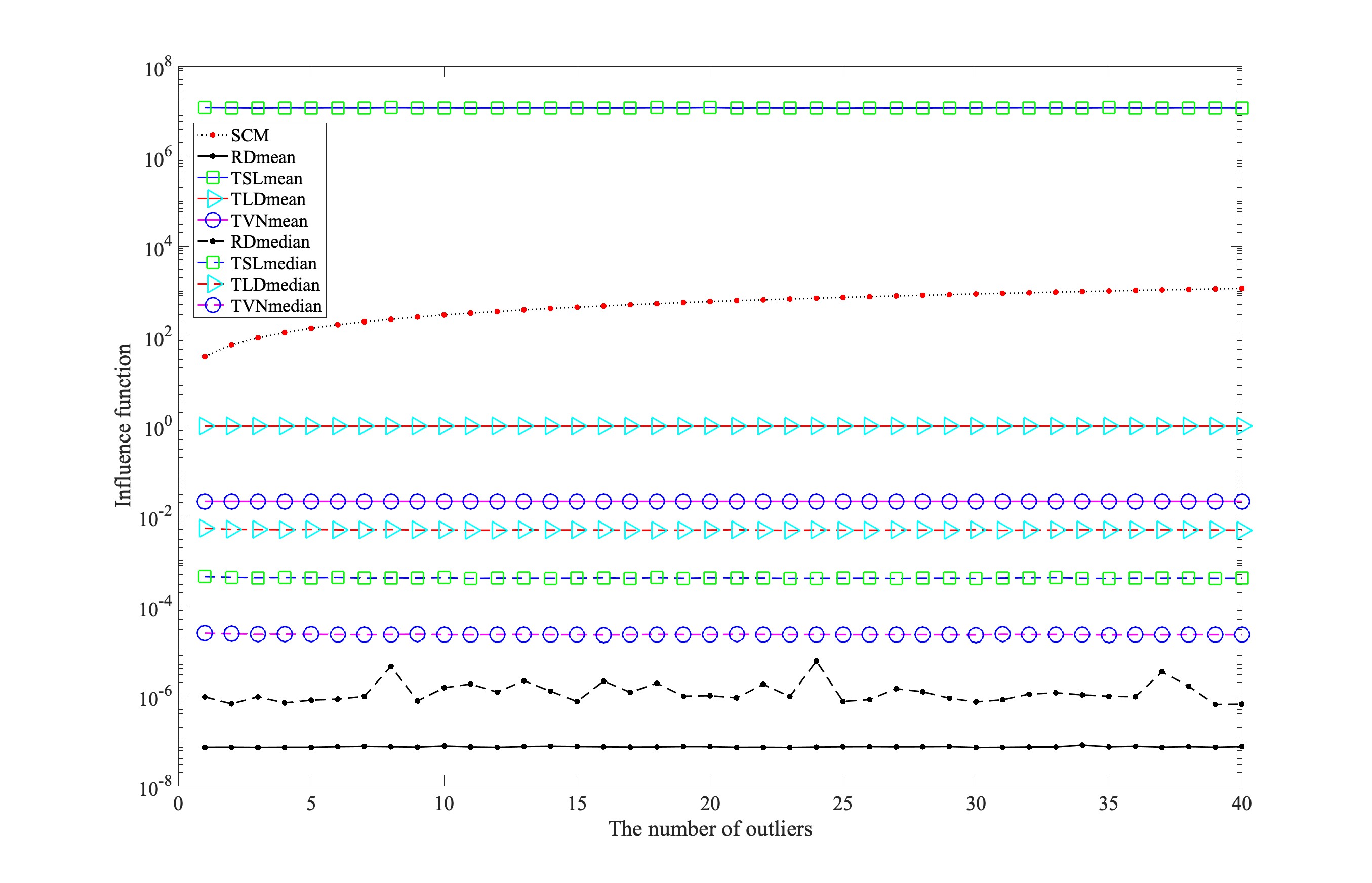}
	\caption{Robustness analysis of the SCM,  the RD mean and median,  and the TBD means and medians via the accurate influence functions.}
	\label{fig:Influencefunction}
\end{figure}

\section{Conclusions}
\label{sec:con}
The TBD medians were defined and applied to matrix-CFAR detectors corresponding to the well-known TBDs, i.e., the TSL, and the TLD and TVN divergences. To illustrate the improvement of detection performance, they were compared with their mean counterparts, the RD mean and median associated to the AIRM, as well as the conventional GLRT detector and the ANMF. 
Simulations showed that the TLD and TVN medians performed better than all means regardless of the number of observations.
To analyze the robustness of those detectors to outliers, we transformed the intractable governing matrix equation into a linear system by introducing an orthogonal basis for Hermitian matrices. This allowed us to calculate the corresponding influence functions analytically and made the robustness analysis accurate and precise. We believe this methodology works equally well for other geometric detectors in matrix-CFAR.
By simulations, the RD mean is most robust among all of them, while  the TBD medians are more robust  compared with the TBD means. From an overall point of view, the TBD medians, in particular the TVN median, have best performance and satisfactory robustness.

As the optimization problem is convex and defined in a convex space, the TBD means exist uniquely.  Unfortunately, an analytical verification of the existence and uniqueness of TBD medians is yet to be constructed, although numerical simulations suggest so.  Future research in applications includes   detection problems for real-world data using the TBD medians. Finally, it would also be interesting to extend the current study to other covariance estimators, such as the diagonal loading \cite{8450037,4383590}, and structured
covariance interference 
(e.g., \cite{COP2016}).


\section*{Acknowledgements.} The authors thank Prof. Sergey Yurish for the kind invitation and for the hospitality during the ASPAI'2021 Conference, in which part of the work was presented \cite{onoproceedings}. 
The authors would also like to thank Xiaoqiang Hua and Hiroki Sato for helpful discussions and to the referees for the constructive comments. This work was partially supported by Japan Society for the Promotion of Science KAKENHI (No. JP20K14365), Japan Science and Technology Agency CREST (No. JPMJCR1914), and the Fukuzawa Fund of Keio University.

\bibliographystyle{plain} 
\bibliography{mybibfile}

\begin{thebibliography}{10}

\bibitem{4383590}
Y.~I. {Abramovich}, N.~K. {Spencer}, and A.~Y. {Gorokhov}.
\newblock Modified {GLRT} and {AMF} framework for adaptive detectors.
\newblock {\em IEEE Transactions on Aerospace and Electronic Systems},
  43(3):1017--1051, 2007.

\bibitem{ABY2012}
M.~Arnaudon, F.~Barbaresco, and L.~Yang.
\newblock Medians and means in {Riemannian} geometry: {Existence}, uniqueness
  and computation.
\newblock In R.~Bhatia and F.~Nielsen, editors, {\em Matrix Information
  Geometry}, pages 169--198, 2012.

\bibitem{6514112}
M.~Arnaudon, F.~Barbaresco, and L.~Yang.
\newblock Riemannian medians and means with applications to radar signal
  processing.
\newblock {\em IEEE Journal of Selected Topics in Signal Processing},
  7(4):595--604, 2013.

\bibitem{5484507}
F.~{Bandiera}, O.~{Besson}, and G.~{Ricci}.
\newblock Knowledge-aided covariance matrix estimation and adaptive detection
  in compound-{G}aussian noise.
\newblock {\em IEEE Transactions on Signal Processing}, 58(10):5391--5396, Oct.
  2010.

\bibitem{BarInoTool}
F.~Barbaresco.
\newblock Innovative tools for radar signal processing based on {C}artan's
  geometry of {SPD} matrices $\&$ information geometry.
\newblock In {\em 2008 IEEE International Radar Conference}, pages 1--6, 2008.

\bibitem{ETIVC2009}
F.~Barbaresco.
\newblock Interactions between symmetric cone and information geometries:
  Bruhat-tits and siegel spaces models for high resolution autoregressive
  doppler imagery.
\newblock In F.~Nielsen, editor, {\em Emerging Trends In Visual Computing},
  pages 124--163, Palaiseau, France, 2009.

\bibitem{BARBARESCO201054}
F.~Barbaresco and U.~Meier.
\newblock Radar monitoring of a wake vortex: Electromagnetic reflection of wake
  turbulence in clear air.
\newblock {\em Comptes Rendus Physique}, 11(1):54--67, 2010.

\bibitem{4154721}
O.~{Besson}, J.~{Tourneret}, and S.~{Bidon}.
\newblock Knowledge-aided {B}ayesian detection in heterogeneous environments.
\newblock {\em IEEE Signal Processing Letters}, 14(5):355--358, May 2007.

\bibitem{40004034}
Y.~Cabanes, F.~Barbaresco, M.~Arnaudon, and J.~Bigot.
\newblock Toeplitz {H}ermitian positive definite matrix machine learning based
  on {F}isher metric.
\newblock In {\em Geometric Science of Information: 4th International
  Conference, GSI 2019, Toulouse, France, August 27-29, 2019, Proceedings},
  Geometric Science of Information, pages 261--270. Springer International
  Publishing, 2019.

\bibitem{DroneDetect}
H.~Chahrour, R.~M. Dansereau, S.~Rajan, and B.~Balaji.
\newblock Target detection through {R}iemannian geometric approach with
  application to drone detection.
\newblock {\em IEEE Access}, 9:123950--123963, 2021.

\bibitem{6556702}
M.~Charfi, Z.~Chebbi, M.~Moakher, and B.~C. Vemuri.
\newblock Bhattacharyya median of symmetric positive-definite matrices and
  application to the denoising of diffusion-tensor fields.
\newblock In {\em 2013 IEEE 10th International Symposium on Biomedical
  Imaging}, pages 1227--1230, 2013.

\bibitem{COP2016}
D.~Ciuonzo, D.~Orlando, and L.~Pallotta.
\newblock On the maximal invariant statistic for adaptive radar detection in
  partially homogeneous disturbance with persymmetric covariance.
\newblock {\em IEEE Signal Processing Letters}, 23(12):1830--1834, Dec 2016.

\bibitem{381910}
E.~Conte, M.~Lops, and G.~Ricci.
\newblock Asymptotically optimum radar detection in compound-{G}aussian
  clutter.
\newblock {\em IEEE Transactions on Aerospace and Electronic Systems},
  31(2):617--625, 1995.

\bibitem{5417154}
A.~{De Maio}, A.~{Farina}, and G.~{Foglia}.
\newblock Knowledge-aided {B}ayesian radar detectors \& their application to
  live data.
\newblock {\em IEEE Transactions on Aerospace and Electronic Systems},
  46(1):170--183, Jan. 2010.

\bibitem{8450037}
A.~{De Maio}, L.~{Pallotta}, J.~{Li}, and P.~{Stoica}.
\newblock Loading factor estimation under affine constraints on the covariance
  eigenvalues with application to radar target detection.
\newblock {\em IEEE Transactions on Aerospace and Electronic Systems},
  55(3):1269--1283, 2019.

\bibitem{7842633}
A.~{Decurninge} and F.~{Barbaresco}.
\newblock Robust {B}urg estimation of radar scatter matrix for autoregressive
  structured {SIRV} based on {F}r{\'e}chet medians.
\newblock {\em IET Radar, Sonar \& Navigation}, 11(1):78--89, 2017.

\bibitem{dhillon2008matrix}
I.~S. Dhillon and J.~A. Tropp.
\newblock Matrix nearness problems with {B}regman divergences.
\newblock {\em SIAM Journal on Matrix Analysis and Applications},
  29(4):1120--1146, 2008.

\bibitem{9054969}
X.~Du, A.~Aubry, A.~De~Maio, and G.~Cui.
\newblock Toeplitz structured covariance matrix estimation for radar
  applications.
\newblock {\em IEEE Signal Processing Letters}, 27:595--599, 2020.

\bibitem{Robuststatistics}
P.~T. Fletcher, S.~Venkatasubramanian, and S.~Joshi.
\newblock Robust statistics on {R}iemannian manifolds via the geometric median.
\newblock In {\em 2008 IEEE Conference on Computer Vision and Pattern
  Recognition}, pages 1--8. IEEE Computer Society, 2008.

\bibitem{SCMbyMLE}
N.~R. Goodman.
\newblock Statistical analysis based on a certain multivariate complex
  {G}aussian distribution.
\newblock {\em The Annals of Mathematical Statistics}, 34(1):152--177, 1963.

\bibitem{Hig2008}
N.~J. Higham.
\newblock {\em Functions of Matrices: Theory and Computation}.
\newblock SIAM, Philadelphia, 2008.

\bibitem{fixpointHua}
X.~Hua, Y.~Cheng, H.~Wang, Y.~Qin, and Y.~Li.
\newblock Geometric means and medians with applications to target detection.
\newblock {\em IET Signal Processing}, 11(6):711--720, 2017.

\bibitem{MatrixCFARbyKL}
X.~Hua, Y.~Cheng, H.~Wang, Y.~Qin, Y.~Li, and W.~Zhang.
\newblock Matrix {CFAR} detectors based on symmetrized {K}ullback--{L}eibler
  and total {K}ullback--{L}eibler divergences.
\newblock {\em Digital Signal Processing}, 69:106--116, 2017.

\bibitem{CFARJensen}
X.~Hua, H.~Fan, Y.~Cheng, H.~Wang, and Y.~Qin.
\newblock Information geometry for radar target detection with total
  {J}ensen--{B}regman divergence.
\newblock {\em Entropy}, 20(4):256, 2018.

\bibitem{myarticle}
X.~Hua, Y.~Ono, L.~Peng, Y.~Cheng, and H.~Wang.
\newblock Target detection within nonhomogeneous clutter via total {B}regman
  divergence-based matrix information geometry detectors.
\newblock {\em IEEE Transactions on Signal Processing}, 69:4326--4340, 2021.

\bibitem{9764734}
X.~Hua, Y.~Ono, L.~Peng, and Y.~Xu.
\newblock Unsupervised learning discriminative {MIG} detectors in
  nonhomogeneous clutter.
\newblock {\em IEEE Transactions on Communications}, 70:4107--4120, 2022.

\bibitem{GLRT}
E.~J. {Kelly}.
\newblock Adaptive detection in non-stationary interference, part {I} and part
  {II}.
\newblock Technical report, Lincoln Laboratory, MIT, June 25, 1985.

\bibitem{BarHFandXband}
J.~{Lapuyade-Lahorgue} and F.~{Barbaresco}.
\newblock Radar detection using {S}iegel distance between autoregressive
  processes, application to {HF} and {X}-band radar.
\newblock In {\em 2008 IEEE Radar Conference}, pages 1--6, May 2008.

\bibitem{LI20191}
N.~Li, H.~Yang, G.~Cui, L.~Kong, and Q.~H. Liu.
\newblock {Adaptive two-step Bayesian MIMO detectors in compound-Gaussian
  clutter}.
\newblock {\em Signal Processing}, 161:1--13, 2019.

\bibitem{s21072391}
M.~Martorella, S.~Gelli, and A.~Bacci.
\newblock Ground moving target imaging via {SDAP}-{ISAR} processing: Review and
  new trends.
\newblock {\em Sensors}, 21(7):2391, 2021.

\bibitem{Moa2005}
M.~Moakher.
\newblock A differential geometric approach to the geometric mean of symmetric
  positive-definite matrices.
\newblock {\em SIAM Journal on Matrix Analysis and Applications},
  26(3):735--747, 2005.

\bibitem{MoaAvg}
M.~Moakher.
\newblock On the averaging of symmetric positive-definite tensors.
\newblock {\em Journal of Elasticity}, 82(3):273--296, 2006.

\bibitem{RiegeoMoa}
M.~Moakher and M.~Z\'era\"i.
\newblock The {R}iemannian geometry of the space of positive-definite matrices
  and its application to the regularization of positive-definite matrix-valued
  data.
\newblock {\em Journal of Mathematical Imaging and Vision}, 40(2):171--187,
  2011.

\bibitem{onoproceedings}
Y.~Ono, L.~Peng, and H.~Sato.
\newblock The matrix-{CFAR} via total {B}regman divergence medians in signal
  detection.
\newblock In S.~Yurish, editor, {\em 3rd International Conference on Advances
  in Signal Processing and Artificial Intelligence}, pages 18--22.
  International Frequency Sensor Association Publishing, 2021.

\bibitem{4101326}
I.~S. Reed, J.~D. Mallett, and L.~E. Brennan.
\newblock Rapid convergence rate in adaptive arrays.
\newblock {\em IEEE Transactions on Aerospace and Electronic Systems},
  AES-10(6):853--863, 1974.

\bibitem{RadarSignalProcessing}
M.~Richards.
\newblock {\em Fundamentals of Radar Signal Processing}.
\newblock McGraw-Hill Professional Professional, 2nd edition, 2014.

\bibitem{Smith1993}
S.~T. Smith.
\newblock {\em Geometric Optimization Methods for Adaptive Filtering}.
\newblock PhD Thesis, Harvard University, Cambridge, Massachusetts, 1993.

\bibitem{Udr1994}
C.~Udri\c{s}te.
\newblock {\em Convex Functions and Optimization Methods on Riemannian
  Manifolds}.
\newblock Springer Science+Business Media, B.V., Dordrecht, 1994.

\bibitem{6853408}
P.~Wang, Z.~Wang, H.~Li, and B.~Himed.
\newblock Knowledge-aided parametric adaptive matched filter with automatic
  combining for covariance estimation.
\newblock {\em IEEE Transactions on Signal Processing}, 62(18):4713--4722,
  2014.

\bibitem{seaclutter}
K.~{Ward}, R.~{Tough}, and S.~{Watts}.
\newblock {\em Sea Clutter: Scattering, the K Distribution and Radar
  Performance}.
\newblock The institution of Engineering and Technology, London, 2013.

\bibitem{ctx11663242640004034}
Z.~Yang, X.~Li, H.~Wang, and R.~Fa.
\newblock {Knowledge-aided STAP with sparse-recovery by exploiting
  spatio-temporal sparsity}.
\newblock {\em IET Signal Processing}, 10(2):150--161, 2016.

\bibitem{7560351}
L.~Ye, Q.~Yang, and W.~Deng.
\newblock Matrix constant false alarm rate ({MCFAR}) detector based on
  information geometry.
\newblock In {\em 2016 IEEE Information Technology, Networking, Electronic and
  Automation Control Conference}, pages 211--215, 2016.

\bibitem{987665}
M.~A. Zatman.
\newblock {ABF} limitations when using either {T}oeplitz covariance matrix
  estimators or the parametric vector {AR} technique.
\newblock In {\em Conference Record of Thirty-Fifth Asilomar Conference on
  Signals, Systems and Computers}, volume~2, pages 1111--1115, 2001.

\end{thebibliography}

\end{document}